\documentclass[a4paper,11pt]{article}
\usepackage[utf8]{inputenc}

\usepackage{amsmath,amssymb,amsthm}
\usepackage{authblk}

\newcommand{\F}{\mathbb{F}}

\newcommand{\N}{\mathbb{N}}

\newcommand{\rk}{\mathrm{rank}}
\newcommand{\GL}{\mathrm{GL}}

\newtheorem{lemma}{Lemma}[section]
\newtheorem{theorem}[lemma]{Theorem}
\newtheorem{proposition}[lemma]{Proposition}
\newtheorem{corollary}[lemma]{Corollary}
\newtheorem{definition}[lemma]{Definition}
\newtheorem{example}[lemma]{Example}

\title{New Criteria for MRD and Gabidulin Codes\\ and some Rank-Metric Code Constructions}

\author{Anna-Lena Horlemann-Trautmann$^*$}
\affil{Algorithmics Laboratory, EPF Lausanne, Switzerland}

\author{Kyle Marshall\thanks{The authors were partially supported by SNF grant no.\ 149716.}}
\affil{Institute of Mathematics, University of Zurich, Switzerland}

\date{}

\begin{document}

\maketitle

\section{Introduction}

Codes in the rank metric have been studied for the last four decades. For linear codes a Singleton-type bound can be derived for these codes. In analogy to MDS codes in the Hamming metric, we call rank-metric codes that achieve the Singleton-type bound MRD (maximum rank distance) codes. Since the works of Delsarte \cite{de78} and Gabidulin \cite{ga85a} we know that linear MRD codes exist for any set of parameters.  The codes they describe are called Gabidulin codes. 

Moreover, Berger in  \cite{be03} and Morrison in \cite{mo14} showed what the linear and semi-linear isometries of rank-metric codes are. It is an open question if there are other general constructions of MRD codes that are not equivalent (under the isometries) to Gabidulin codes. Recently several results have been established in this direction, e.g.\ in \cite{co15,cr15,sh15}, where many of the derived codes are not linear over the underlying field but only linear over some subfield of it. Hence it is still an open question to find other constructions of non-Gabidulin MRD codes.

In this paper we want to derive criteria for MRD and Gabidulin codes and use these to come up with new non-Gabidulin MRD codes that are linear over the original field, not only a subfield. Moreover, we want to give a classification of these codes and investigate how many different equivalence classes of MRD codes we get for small parameters.

This paper is structured as follows. In Section \ref{sec:preliminaries} we give some preliminaries on finite fields, rank-metric codes and Gabidulin codes. In Section \ref{sec:criterionMRD} we present a new criterion for MRD codes, in Section \ref{sec:criterionGab} we derive a criterion for Gabidulin codes. In Section \ref{sec:constructions} we use the results of Sections \ref{sec:criterionMRD} and \ref{sec:criterionGab} to find new non-Gabidulin MRD codes for small parameters. We conclude in Section \ref{sec:conclusion}.

%%%%%%%%%%%%%%%%%%%%%%%%%%%%%%%%%%%%%%%%%%%%%%%%%%%%%%%%%%%%%%%%%%%%%%%%%%%%%%%%%%%%%%%%%%%%%%%%%%%%%%%%%%%%%%%%%%%%%%%%%%%%%%%%%%%%%%%%

\section{Preliminaries}\label{sec:preliminaries}

Let $q$ be a prime power and let $\F_q$ denote the finite field with $q$ elements. It is well-known that there always exists a primitive element $\alpha$ of the extension field $\F_{q^m}$, such that $\F_{q^m}\cong \F_q[\alpha] $. Moreover, $\F_{q^m}$  is isomorphic (as a vector space over $\F_q$) to the vector space $\F_q^m$. 
%If not noted differently we will use the isomorphism
%\begin{align*}
%\F_q^m &\longrightarrow \F_{q^m}\cong \F_q[\alpha]  \\
%(v_1, \dots, v_m) &\longmapsto \sum_{i=1}^m v_i \alpha^{i-1}    .
%\end{align*}
One then easily obtains the isomorphic description of matrices over the base field $\F_q$ as vectors over the extension field, i.e.\ $\F_q^{m\times n}\cong \F_{q^m}^n$. Since we will work with matrices over different underlying fields, we denote the rank of a matrix $X$ over $\F_q$ by $\rk_q(X)$.

\begin{definition}
 The \emph{rank distance} $d_R$ on  $\F_q^{m\times n}$ is defined by
\[d_R(X,Y):= \rk_q(X-Y) , \quad X,Y \in \F_q^{m\times n}. \]
Analogously, we define the rank distance between two elements $\boldsymbol x,\boldsymbol y \in \F_{q^m}^n$ as the rank of the difference of the respective matrix representations in $\F_q^{m\times n}$.
\end{definition}

In this paper we will focus on $ \F_{q^m}$-linear rank-metric codes in $\F_{q^m}^n$, i.e.\ those codes that form a vector space over  $ \F_{q^m}$. Whenever we talk about linear codes in this work, we will mean linearity over the extension field $ \F_{q^m}$. 
The well-known Singleton bound for codes in the Hamming metric implies also an upper bound for codes in the rank metric:
%We have the following upper bound on the dimension of such codes.
\begin{theorem}\cite[Section~2]{ga85a}
 Let $\mathcal{C}\subseteq \F_{q^m}^{n}$ be a linear matrix code with minimum rank distance $d$ of dimension $k$ (over $\F_{q^m}$). Then 
$$ k\leq n-d+1  .$$
\end{theorem}

\begin{definition}
 A code attaining the Singleton bound is called a \emph{maximum rank distance (MRD) code}.
\end{definition}

For some vector $(v_1,\dots, v_n) \in \F_{q^m}^n$ we denote the $k \times n$ \emph{Moore matrix} by
\[M_k(v_1,\dots, v_n) := \left( \begin{array}{cccc}  v_1 & v_2 &\dots &v_n \\ v_1^{[1]} & v_2^{[1]} &\dots &v_n^{[1]} \\ &&\vdots \\  v_1^{[k-1]} & v_2^{[k-1]} &\dots &v_n^{[k-1]} \end{array}\right)   ,\]
where $[i]:= q^i$.
\begin{definition}\label{def:Gab}
Let $g_1,\dots, g_n \in \F_{q^m}$ be linearly independent over $\F_q$. We define a \emph{Gabidulin code} $\mathcal{C}\subseteq \F_{q^m}^{n}$ of dimension $k$ as the linear block code with generator matrix $M_k(g_1,\dots, g_n)$.        
Using the isomorphic matrix representation we can interpret $\mathcal{C}$ as a matrix code in $\F_q^{m\times n}$.
\end{definition}

\begin{theorem}\cite[Section4]{ga85a}
A Gabidulin code $\mathcal{C}\subseteq \F_{q^m}^{n}$ of dimension $k$ over $\F_{q^m}$ has minimum rank distance %(over $\F_q$)
 $n-k+1$. Thus Gabidulin codes are MRD codes.
\end{theorem}

The dual code of a code $\mathcal{C}\subseteq \F_{q^{m}}^{n}$ is defined in the usual way as 
\[\mathcal{C}^{\perp} := \{u \in \F_{q^{m}}^{n} \mid \boldsymbol{u}\boldsymbol{c}^T=0 \quad \forall \boldsymbol{c}\in \mathcal{C}\}.  \]
%It is well-known that if $\dim(\mathcal{C})=k$, then $\dim(\mathcal{C}^{\perp})=n-k$.
In his seminal paper Gabidulin showed the following two results on dual codes of MRD codes:
\begin{proposition}\cite[Sections~2 and 4]{ga85a}\label{prop:dual1}
\begin{enumerate}
\item
Let $\mathcal{C}\subseteq \F_{q^{m}}^{n}$ be an MRD code of dimension $k$. Then the dual code $\mathcal{C}^{\perp}\subseteq \F_{q^{m}}^{n}$ is an MRD code of dimension $n-k$.
\item
Let $\mathcal{C}\subseteq \F_{q^{m}}^{n}$ be a Gabidulin code of dimension $k$. Then the dual code $\mathcal{C}^{\perp}\subseteq \F_{q^{m}}^{n}$ is a Gabidulin code of dimension $n-k$.
\end{enumerate}
\end{proposition}
Note that the second result in Proposition \ref{prop:dual1} was not stated like this in \cite{ga85a}; Gabidulin showed however that the parity check matrix 
%$H\in \F_{q^m}^{(n-k)\times n}$
 of a Gabidulin code is of the form described in Definition \ref{def:Gab}, which implies the statement. 
For more information on bounds and constructions of rank-metric codes the interested reader is referred to \cite{ga85a}.

The results of Gabidulin (and Delsarte) were later on generalized by Kshevetskiy and Gabidulin in \cite{ks05} as follows.
 \begin{definition}
Let $g_1,\dots, g_n \in \F_{q^m}$ be linearly independent over $\F_q$ and $s\in \N$ such that $\gcd(s,m)=1$. We define a \emph{generalized Gabidulin code} $\mathcal{C}\subseteq \F_{q^m}^{n}$ as the linear block code with generator matrix
\[\left( \begin{array}{cccc}  g_1 & g_2 &\dots &g_n \\ g_1^{[s]} & g_2^{[s]} &\dots &g_n^{[s]} \\ &&\vdots \\  g_1^{[s(k-1)]} & g_2^{[s(k-1)]} &\dots &g_n^{[s(k-1)]} \end{array}\right)  .\]        
\end{definition}

\begin{theorem}\cite[Subsection IV.C]{ks05}
A generalized Gabidulin code $\mathcal{C}\subseteq \F_{q^m}^{n}$ of dimension $k$ over $\F_{q^m}$ has minimum rank distance %(over $\F_q$)
 $n-k+1$. Thus generalized Gabidulin codes are MRD codes.
\end{theorem}

Similarly to the non-generalized case, Kshevetskiy and Gabidulin also showed the following.
\begin{proposition}\cite[Subsection IV.C]{ks05}\label{prop:dual2}
Let $\mathcal{C}\subseteq \F_{q^{m}}^{n}$ be a generalized Gabidulin code of dimension $k$. Then the dual code $\mathcal{C}^{\perp}\subseteq \F_{q^{m}}^{n}$ is a generalized Gabidulin code of dimension $n-k$.
\end{proposition}

Denote by $\mathrm{Gal}(\F_{q^m}/\F_q)$ the \emph{Galois group} of $\F_{q^m}$, i.e.\ the automorphisms of $\F_{q^m}$ that fix the base field $\F_q$. It is well-known that $\mathrm{Gal}(\F_{q^m}/\F_q)$ is generated by the \emph{Frobenius map}, which takes an element to its $q$-th power. Hence the automorphisms are of the form $x\mapsto x^{[i]}$ for some $0\leq i \leq m$. We will denote the respective inverse map, i.e.\ the $[i]$-th root, by $x\mapsto x^{[-i]}$. 

The (semi-)linear rank isometries on $\F_{q^m}^{n}$ are induced by the isometries on $\F_q^{m\times n}$ and are hence well-known, see e.g.\ \cite{be03,mo14,wa96}:
\begin{lemma}\cite[Proposition~2]{mo14}\label{isometries}
The semilinear $\F_q$-rank isometries on $\F_{q^m}^{n}$ are of the form
\[(\lambda, A, \sigma) \in \left( \F_{q^m}^* \times  \GL_n(q) \right) \rtimes \mathrm{Gal}(\F_{q^m}/\F_q) ,\]
acting on $ \F_{q^m}^n \ni (v_1,\dots,v_n)$ via
\[(v_1,\dots,v_n) (\lambda, A, \sigma) = (\sigma(\lambda v_1),\dots,\sigma(\lambda v_n)) A .\]
In particular, if $\mathcal{C}\subseteq \F_{q^m}^n$ is a linear code with minimum rank distance $d$, then 
\[\mathcal{C}' = \sigma(\lambda \mathcal{C}) A \]
is a linear code with minimum rank distance $d$.
%Since we want to study isometry classes of linear codes we will neglect $A$ and $H$.
\end{lemma}

%Throughout the paper we apply the Frobenius map to field elements, vectors, matrices and subspaces, where we always mean to take the Frobenius element-wise.

We denote by $\GL_n(q):=\{A\in \F_q^{n\times n} \mid \rk (A) =n\}$ the general linear group of degree $n$ over $\F_q$. 
One can easily check that $\F_q$-linearly independent elements in $\F_{q^m}$ remain $\F_q$-linearly independent under the actions of $\F_{q^m}^*, \GL_n(q)$ and $\mathrm{Gal}(\F_{q^m}/\F_q)$. Moreover, the Moore matrix structure is preserved under these actions, which implies that the class of Gabidulin codes is closed under the semilinear isometries.

In this work we want to classify MRD codes and which of them are generalized Gabidulin codes. For this we will derive some criteria for both the MRD and the Gabidulin property. 
The following criterion for MRD codes was already given in \cite{ga85a}:
\begin{proposition}\label{lem3}
 Let $H\in \F_{q^m}^{(n-k)\times n}$ be a parity check matrix of a rank-metric code $\mathcal{C}\subseteq \F_{q^m}^n$. Then $\mathcal{C}$ is an MRD code if and only if 
$$ \rk_{q^m}(VH^T) =n-k$$
for all $V\in \F_q^{(n-k)\times n}$ with $\rk_{q}(V)=n-k$.
\end{proposition}

This criterion is formulated with respect to the parity check matrix of a linear code. We can easily derive a criterion for the generator matrix of MRD codes from this:
\begin{corollary}\label{cor3}
  Let $G\in \F_{q^m}^{k\times n}$ be a generator matrix of a rank-metric code $\mathcal{C}\subseteq \F_{q^m}^n$. Then $\mathcal{C}$ is an MRD code if and only if 
$$ \rk_{q^m}(VG^T) =k$$
for all $V\in \F_q^{k\times n}$ with $\rk_{q}(V)=k$.
\end{corollary}
\begin{proof}
 The generator matrix $G$ of $\mathcal{C}$ is a parity check matrix of the dual code $\mathcal{C}^\perp \subseteq \F_{q^m}^n$ of dimension $n-k$. It follows from Proposition \ref{lem3} that $\mathcal{C}^\perp$ is an MRD code if and only if 
$ \rk_{q^m}(VG^T) =k$
for all $V\in \F_q^{k\times n}$ with $\rk_{q}(V)=k$. Since $\mathcal{C}$ is MRD if and only if $\mathcal{C}^\perp$ is MRD (see Proposition \ref{prop:dual1}), the statement follows.
\end{proof}

Throughout the paper $I_k$ denotes the identity matrix of size $k$. Furthermore, $\langle v_1, \dots, v_n \rangle_{q}$ denotes the $\F_q$-vector space generated by $v_1, \dots, v_n$.

%%%%%%%%%%%%%%%%%%%%%%%%%%%%%%%%%%%%%%%%%%%%%%%%%%%%%%%%%%%%%%%%%%%%%%%%%%%%%%%%%%%%%%%%%%%%%%%%%%%%%%%%%%%%%%%%%%%%%%%%%%%%%%%%%%%%%%%%

\section{New Criterion for MRD Codes}\label{sec:criterionMRD}

In this section we give a new criterion to check if a given generator matrix $G$ generates an MRD code. 
%, in contrast to Lemma \ref{lem3} where the \emph{parity check matrix} $H$ (with $GH^T=0$) is used to check the MRD property of a code.
 The criterion is stated in Theorem \ref{thm:mainMRD}. Before we can state the main theorem we need the following lemma.

%\begin{lemma}\label{lem:helpMRD2}
% Any MRD code $C\subseteq \F_{q^m}^n$ of dimension $k$ has a generator matrix $G\in \F_{q^m}^{k\times n}$ in systematic form, i.e.\ 
%\[ G = \left[\; I_k \mid * \; \right] .\]
%Moreover, all entries of $*$ are from $\F_{q^m}\backslash \F_q$.
%\end{lemma}
%\begin{proof}
% Assume that the generator matrix of $C$ in reduced row echelon form has a row with pivot in column $i>k$. Then this row vector has at most $n-k$ many non-zero entries, which contradicts the minimum rank distance $n-k+1$ of $C$. Therefore, all pivots are in columns $1,\dots, k$, which proves the first statement. The second statement follows again from the minimum rank distance $n-k+1$ of the code, because every codeword needs to have at least $n-k$ entries from $\F_{q^m}\backslash \F_q$.
%\end{proof}

\begin{lemma}\label{lem:helpMRD}
 Any generator matrix $G\in \F_{q^m}^{k\times n}$ of an MRD code $\mathcal{C}\subseteq \F_{q^m}^n$ of dimension $k$ has only non-zero maximal minors. 
\end{lemma}
\begin{proof}
Let $V=[\: I_{k} \mid 0_{k\times (n-k)} \: ] \in \F_{q}^{k\times n}$. Then $\det(VG^{T})$ is the maximal minor of $G$ involving the first $k$ columns. By Corollary \ref{cor3} this minor is non-zero. Similarly we can create all other maximal minors of $G$ by multiplication with some  $ V\in \F_{q}^{k\times n}$ on the left, which implies, by Corollary \ref{cor3}, the statement.
\end{proof}

We can now state the new MRD criterion:

\begin{theorem}\label{thm:mainMRD}
 Let $G\in \F_{q^m}^{k\times n}$ be a generator matrix of a rank-metric code $\mathcal{C}\subseteq \F_{q^m}^n$. Then $\mathcal{C}$ is an MRD code if an only if for any $A\in \GL_{n}(q)$, every maximal minor of $GA$ is non-zero.  
\end{theorem}
\begin{proof}
We first prove the \emph{only if} direction. For this let $\mathcal{C}$ be MRD. Then we know from Lemma \ref{isometries} that all elements on the orbit of $\mathcal{C}$ under $\GL_n(q)$ are MRD. Since $\GL_n(q)$ acts on the columns of any generator matrix of $\mathcal{C}$, together with Lemma \ref{lem:helpMRD}, we get that all maximal minors of any orbit element must be non-zero.

For the other direction, let $\mathcal{C}$ be non-MRD, i.e.\ there exists a non-zero codeword $\boldsymbol{c} \in \mathcal{C}$ of rank at most $n-k$. Then there exists $A\in \GL_n(q)$ s.t.\ 
\[cA = ( \underbrace{0 \dots 0}_{k} \mid \underbrace{* \dots *}_{n-k} ) .\]
This in turn implies that there exists a generator matrix of $\mathcal{C}A$ with $cA$ as a row. Thus the first maximal minor of this generator matrix will be zero. 
\end{proof}

We can slightly simplify this criterion as follows. For this denote by $\mathrm{UT}^*_{n}(q)$ the subgroup of $\GL_{n}(q)$ of upper triangular matrices with an all-$1$ diagonal.

\begin{corollary}\label{cor:mainMRD}
 Let $G\in \F_{q^m}^{k\times n}$ be a generator matrix of a rank-metric code $\mathcal{C}\subseteq \F_{q^m}^n$. Then $\mathcal{C}$ is an MRD code if and only if for any $A\in \mathrm{UT}^*_{n}(q)$ every maximal minor of $GA$ is non-zero.  
\end{corollary}
\begin{proof}
 Note that $\mathrm{UT}^*_{n}(q)$, together with the diagonal matrices and the permutation matrices in $\GL_n(q)$ generate the whole general linear group $\GL_n(q)$. The action of the diagonal matrices multiplies the maximal minors of the generator matrix by a non-zero scalar, the action of the permutation matrices at most changes the sign of the maximal minors. Hence, these two subgroups do not change the \emph{non-zero-ness} of the maximal minors.
\end{proof}

%%%%%%%%%%%%%%%%%%%%%%%%%%%%%%%%%%%%%%%%%%%%%%%%%%%%%%%%%%%%%%%%%%%%%%%%%%%%%%%%%%%%%%%%%%%%%%%%%%%%%%%%%%%%%%%%%%%%%%%%%%%%%%%%%%%%%%%%

\section{New Criterion for Gabidulin Codes}\label{sec:criterionGab}

In this section we derive a new criterion to establish if a given MRD code is a generalized Gabidulin code or not. The main result is stated in Theorem~\ref{thm:mainGab}.

Let $\mathcal{C}\subseteq \F_{q^m}^n$ be a linear MRD code of dimension $k$ (i.e.\ rank distance $d=n-k+1$) with generator matrix $G$. 
%By $\mathcal{C}^{(q)}$ and $G^{(q)}$ we denote the coordinate-wise Frobenius map. 
Recall the notation $[i]:=q^i$. We apply the Frobenius on vectors and matrices coordinate-wise, i.e., for $G\in \F_{q^m}^{k\times n}$ we have $G^{[i]} = (g_{jk}^{[i]})_{j,k}$ and for $\mathcal{C}\subseteq \F_{q^m}^n$ we have 
$\mathcal{C}^{[i]} = \{\boldsymbol c^{[i]} \mid \boldsymbol c \in \mathcal{C}\}$. 
In this section we let $s\in \N$, $s<m$ be such that $\gcd(s,m)=1$.

The following three Lemmas are needed to prove Proposition \ref{prop:Gab} and then Theorem \ref{thm:mainGab}.

\begin{lemma}\label{lem:inverseFrob}
Let $A\in \mathrm{GL}_{k}(q^{m})$. Then $(A^{-1})^{[1]} = (A^{[1]})^{-1}$.
\end{lemma}
\begin{proof}
We have that
\begin{align*} A^{-1} A = I_{k} &\iff (A^{-1} A)^{[1]} = I_{k} \\&\iff (A^{-1})^{[1]} A^{[1]} = I_{k} \\&\iff  (A^{-1})^{[1]} =(A^{[1]})^{-1} .
\end{align*}
\end{proof}

It is well-known that the roots of $x^q-x$ are exactly the elements of $\F_q$ (see e.g. \cite[Theorem 2.5]{li94}). For our main results we need a generalization of this result:
\begin{lemma}\label{lem:help}
If $\gcd(s,m)=1$, then the roots in $\F_{q^m}$ of $x^{[s]}-x$ are exactly the elements of $\F_q$.
\end{lemma}
\begin{proof}
 Consider the field $F_{q^{ms}}$, then both $\F_{q^{m}}$ and $\F_{q^{s}}$ are subfields of it \cite[Theorem 2.6]{li94}.  Since $m$ and $s$ are coprime these two subfields only intersect in the base field $\F_{q}$.  Moreover, the roots of $x^{[s]}-x$ in $\F_{q^{ms}}$ are exactly the elements of $\F_{q^{s}}$, hence  the roots of it in $\F_{q^{m}}$ are the elements of $\F_{q}$.
\end{proof}

%\begin{lemma}\label{lem:indep}
% Let $v=(v_1,\dots,v_n)\in \F_{q^m}^n$ be of rank $r$ over $\F_q$. Then $v,v^{[1]},\dots, v^{[r-1]}$ are linearly independent over $\F_{q^{m}}$.
%\end{lemma}
%\begin{proof}
% Assume that $v,v^{[1]},\dots, v^{[r-1]}$ are not linearly independent over $\F_{q^{m}}$, i.e.\ there exist $\lambda_0,\dots,\lambda_{r-1}\in\F_{q^{m}}$, at least one $\lambda_i\neq 0$, such that 
%\[\sum_{i=0}^{r-1} \lambda_i v^{[i]}   =0 .\]
%Then the $q$-linearized polynomial $p(x):= \sum_{i=0}^{r-1} \lambda_i x^{[i]}$ has roots $v_1,\dots,v_n$. Since $p(x)$ is linearized, all elements of the vector space $\langle v_1,\dots, v_n\rangle_q$ are a root of it. Since  $\langle v_1,\dots, v_n\rangle_q$ has dimension $r$, there are $q^r$ roots. Hence $p(x)$ has degree at least $q^r$, which is a contradiction.
%\end{proof}

\begin{lemma}\label{lem:indep}
 Let $\boldsymbol{v}=(v_1,\dots,v_n)\in \F_{q^m}^n$ be of rank $r$ over $\F_q$. Then $\boldsymbol{v},\boldsymbol{v}^{[s]},\dots, \boldsymbol{v}^{[s(r-1)]}$ are linearly independent over $\F_{q^{m}}$.
\end{lemma}
\begin{proof}
 Assume that $\boldsymbol{v},\boldsymbol{v}^{[s]},\dots, \boldsymbol{v}^{[s(r-1)]}$ are not linearly independent over $\F_{q^{m}}$, i.e.\ there exist $\lambda_0,\dots,\lambda_{r-1}\in\F_{q^{m}}$, at least one $\lambda_i\neq 0$, such that 
\[\sum_{i=0}^{r-1} \lambda_i \boldsymbol{v}^{[is]}   =\boldsymbol{0} .\]
Then the $q^{s}$-linearized polynomial $p(x):= \sum_{i=0}^{r-1} \lambda_i x^{[si]} = \sum_{i=0}^{r-1} \lambda_i x^{(q^s)^{i}}  \in \F_{q^{ms}}[x]$ has roots $v_1,\dots,v_n$. Since $p(x)$ is linearized, all elements of the vector space $\langle v_1,\dots, v_n\rangle_{q^{s}}$ are roots of it. Since  $\langle v_1,\dots, v_n\rangle_q$ has dimension $r$, by \cite[Lemma 4.3]{ks05}, also $\langle v_1,\dots, v_n\rangle_{q^{s}}$ has dimension $r$. Hence, there are $q^{rs}$ roots of $p(x)$ in $\F_{q^{ms}}$. Hence $p(x)$ must have degree at least $q^{rs}$, which is a contradiction.
\end{proof}

The following straight-forward lemma is needed to prove Lemma \ref{lem:main}.
\begin{lemma}\label{lem:indep2}
 Let $\boldsymbol{w}_1,\dots, \boldsymbol{w}_k \in \F_{q^m}^n$ be linearly independent over $\F_{q^m}$. Then  $\boldsymbol{w}_1^{[s]},\dots, \boldsymbol{w}_k^{[s]} \in \F_{q^m}^n$ are also linearly independent over $\F_{q^m}$
\end{lemma}
\begin{proof}
 Assume that $\boldsymbol{w}_1^{[s]},\dots, \boldsymbol{w}_k^{[s]}$ are not linearly independent, i.e.\ there exist $\lambda_1,\dots,\lambda_k \in \F_{q^m}$ with
\[\sum_{i=1}^k \lambda_i \boldsymbol{w}_i^{[s]} = \boldsymbol{0} \iff \left(\sum_{i=1}^k \lambda_i^{[-s]} \boldsymbol{w}_i\right)^{[s]} = \boldsymbol{0}  \iff \sum_{i=1}^k \lambda_i^{[-s]} \boldsymbol{w}_i = \boldsymbol{0} .\]
Thus the vectors $\boldsymbol{w}_1,\dots, \boldsymbol{w}_k$ are not linearly independent over $\F_{q^m}$, which is a contradiction.
\end{proof}

The following result is a generalization of \cite[Theorem 1]{gi10}.

\begin{lemma}\label{lem:main}
Let $\mathcal{W}\subset \F_{q^m}^n$ be a subspace of dimension $k \leq n$ satisfying $\mathcal{W}^{[s]} = \mathcal{W}$. Then $\mathcal{W}$ has a generator matrix in $\F_q^{k\times n}$. In particular $\mathcal{W}$ contains elements of rank $1$ over $\F_{q}$.
\end{lemma}
\begin{proof}
If $\{\boldsymbol{w}_1, \dots , \boldsymbol{w}_{k}\} \subset \F_{q^m}^n$ is a basis for $\mathcal{W}$, then by Lemma \ref{lem:indep2} $\{\boldsymbol{w}_1^{[s]}, \dots , \boldsymbol{w}_k^{[s]}\}$ is also a basis of $\mathcal{W}$. 
Then there exists $A\in \mathrm{GL}_k(q^{m})$ such that 
$$\left(\begin{array}{cccc}w_{1,1}^{[s]} & w_{1,2}^{[s]} & \ldots & w_{1,n}^{[s]} \\ w_{2,1}^{[s]} & w_{2,2}^{[s]} & \ldots & w_{2,n}^{[s]} \\ &  & \vdots &  \\ w_{k,1}^{[s]} & w_{k,2}^{[s]} & \ldots & w_{k,n}^{[s]} \end{array}\right) = A \left(\begin{array}{cccc}w_{1,1} & w_{1,2} & \ldots & w_{1,n} \\ w_{2,1} & w_{2,2} & \ldots & w_{2,n} \\ &  & \vdots &  \\ w_{k,1} & w_{k,2} & \ldots & w_{k,n} \end{array}\right).$$ 
Since the rightmost matrix has rank $k$, there exists a set of $k$ linearly independent (over $\F_{q^{m}}$) columns. 
Without loss of generality, we can assume that the first $k$ columns are linearly independent. Thus the submatrix $W_1 := (w_{i,j})_{i,j=1}^{k}$ is invertible (and therefore $W_1^{[s]}$ is also invertible by Lemma \ref{lem:indep2}), and so we can solve 
$$A = W_1^{[s]}W_1^{-1}.$$ 
Define $W_2 := \left.(w_{i,j})_{i=1}^{k}\right._{j=k+1}^{n}$. Then we have 
$$W_2^{[s]} = W_1^{[s]}W_1^{-1}W_2 .$$ 
If we apply the Frobenius map $s$ times on both sides and use Lemma \ref{lem:inverseFrob}, we obtain 
\begin{align*} 
W_2^{[2s]} &= W_1^{[2s]}(W_1^{-1})^{[s]}W_2^{[s]} \\
 &= W_1^{[2s]}(W_1^{[s]})^{-1}W_1^{[s]}W_1^{-1}W_2 \\ &= W_1^{[2s]}W_{1}^{-1}W_2. 
\end{align*} 
Then, we have 
$$W_1^{[2s]}(W_1^{-1})^{[s]}W_2^{[s]} = W_1^{[2s]}W_{1}^{-1}W_2.$$ 
Since $W_1^{[2s]}$ is invertible, we obtain $$(W_1^{-1}W_2)^{[s]} = W_1^{-1}W_2,$$ and therefore we must have that $W_1^{-1}W_2$ has only entries in $\F_q$, by Lemma \ref{lem:help}. Therefore, a generator matrix for $\mathcal{W}$ can be expressed as $W_{1}^{-1}[W_{1} \mid W_{2}] = [I_k \mid W_1^{-1}W_2] \in \F_q^{k\times n}$, whose rows have rank weight $1$ over $\F_{q}$.
\end{proof}

We can now state and prove the central ingredient for the main result in Theorem \ref{thm:mainGab}:
%\begin{proposition}\label{prop:Gab}
%Suppose that ${C}\subset \F_{q^m}^{n}$ is an MRD code of dimension $k$. 
% If $\dim(C \cap C^{[s]}) = k-1$ (this automatically implies that $k<n$ and hence that the minimum distance is at least $2$), then there exists a generator matrix for $C$ of the form
%\[G^* = \left(\begin{array}{cccc}
%g_1 & g_2 & \dots & g_n\\
%g_1^{[s]} & g_2^{[s]}& \dots & g_n^{[s]}\\
%&&\vdots\\
%g_1^{[s(k-1)]} & g_2^{[s(k-1)]} & \dots & g_n^{[s(k-1)]}
%              \end{array}
% \right)\]
%with $g_1,\dots,g_n \in \F_{q^m}$.
%\end{proposition}

\begin{proposition}\label{prop:Gab}
Suppose that ${\mathcal{C}}\subset \F_{q^m}^{n}$ is a linear code of dimension $k\geq 2$ and minimum rank distance at least $k$. 
 If $\dim(\mathcal{C} \cap \mathcal{C}^{[s]}) = k-1$ (this automatically implies that $k<n$), then there exists a generator matrix for $\mathcal{C}$ of the form
\[G^* = \left(\begin{array}{cccc}
g_1 & g_2 & \dots & g_n\\
g_1^{[s]} & g_2^{[s]}& \dots & g_n^{[s]}\\
&&\vdots\\
g_1^{[s(k-1)]} & g_2^{[s(k-1)]} & \dots & g_n^{[s(k-1)]}
              \end{array}
 \right)\]
with $g_1,\dots,g_n \in \F_{q^m}$.
\end{proposition}
\begin{proof}
 We prove this inductively on $k$. First assume that $k=2$. Then $\dim(\mathcal{C} \cap \mathcal{C}^{[s]}) = 1$, i.e.\ there exists $\boldsymbol{g}'\in \mathcal{C}$ such that $ \mathcal{C} \cap \mathcal{C}^{[s]} = \langle \boldsymbol{g}' \rangle_{q^m}$. Since $\boldsymbol{g}' \in \mathcal{C}^{[s]}$, we get that ${\boldsymbol{g}'}^{[-s]} \in \mathcal{C}$. The minimum rank distance of $\mathcal{C}$ is at least $k=2$, i.e.\ the rank of ${\boldsymbol{g}'}^{[-s]}$ over $\F_{q}$ is at least $2$. Then, by Lemma \ref{lem:indep}, ${\boldsymbol{g}'}^{[-s]}$  and $\boldsymbol{g}'$ are linearly independent. 
Hence they form a basis of $\mathcal{C}$ and we can rename $\boldsymbol{g}:={\boldsymbol{g}'}^{[-s]}$ to write a generator matrix 
\[G^* = \left(\begin{array}{cc}
\boldsymbol{g} \\ \boldsymbol{g}^{[s]}     \end{array}
 \right)
.\]

We now explain the induction step $(k-1)\rightarrow k$.
Let $\mathcal{W} = {\mathcal{C}}\cap {\mathcal{C}}^{[s]}$, then we know from Lemma \ref{lem:main} that $\mathcal{W}^{[s]} \neq \mathcal{W}$, because the minimum rank distance of $\mathcal{C}$ is at least $k$. Since $\mathcal{W},\mathcal{W}^{[s]} \subset \mathcal{C}^{[s]}$, both with codimension $1$, we get $\langle \mathcal{W}, \mathcal{W}^{[s]}\rangle_{q^m} =\mathcal{C}^{[s]}$. Then
$$\dim(\mathcal{W}\cap \mathcal{W}^{[s]}) = \dim(\mathcal{W}) + \dim(\mathcal{W}^{[s]}) - \dim(\mathcal{W}+\mathcal{W}^{[s]}) = 2(k-1) - k = k-2  .$$
Furthermore, since $\mathcal{W}\subset \mathcal{C}$, the minimum rank distance of $\mathcal{W}$ is at least $k$. 
Therefore, $\mathcal{W}$ satisfies the conditions of the induction hypothesis, and so we can express $\mathcal{W}$ in terms of some basis of the form 
$$\{\boldsymbol{w}, \boldsymbol{w}^{[s]}, \dots, \boldsymbol{w}^{[s(k-2)]}\}.$$ 
Hence, $\{\boldsymbol{w}, \boldsymbol{w}^{[s]}, \dots, \boldsymbol{w}^{[s(k-2)]}\}\in \mathcal{C}$ and thus $\{\boldsymbol{w}^{[s]}, \boldsymbol{w}^{[2s]}, \dots, \boldsymbol{w}^{[s(k-1)]}\}\in \mathcal{C}^{[s]}$.  On the other hand, $\boldsymbol{w}\in \mathcal{W}\subset \mathcal{C}^{[s]}$, i.e.\ $\{\boldsymbol{w}, \boldsymbol{w}^{[s]}, \dots, \boldsymbol{w}^{[s(k-1)]}\}\in \mathcal{C}^{[s]}$. By Lemma \ref{lem:indep} this set is linearly independent, i.e.\ it is a basis of $\mathcal{C}^{[s]}$. This in turn implies that $\{\boldsymbol{w}^{[-s]},\boldsymbol{w}, \boldsymbol{w}^{[s]}, \dots, \boldsymbol{w}^{[s(k-2)]}\}$ is a basis of $\mathcal{C}$. Define $\boldsymbol{g}=\boldsymbol{w}^{[-s]}$, then $\{\boldsymbol{g}, \boldsymbol{g}^{[s]}, \dots, \boldsymbol{g}^{[s(k-1)]}\}$ is a basis of $\mathcal{C}$.

%, and $\langle w,\dots, w^{([k-2])}\rangle_{q^m} \subset \mathcal{C}^{(q)}$ is a subspace of dimension $k-1$ and $\langle \boldsymbol{c}^{[1]}, w^{([k-1])} \rangle_{q^m} \subset \mathcal{C}^{(q)}$ is a subspace of dimension $2$. We get
%$$ \dim\left(\langle \boldsymbol{c}^{([1])}, w^{([k-1])}\rangle \cap \langle w, \dots, w^{([k-2])}\rangle \right) = 2+(k-1) - k =1 .$$
% Therefore, there exist scalars $\alpha, \gamma_1, \dots, \gamma_{k-1}$ so that $$\alpha\boldsymbol{c}^{([1])} - \gamma_{k-1}w^{([k-1])} = \sum_{i=0}^{k-2}\gamma_i w^{([i])}.$$ We then obtain, $$\alpha \boldsymbol{c}^{([1])} = \sum_{i=0}^{k-1}\gamma_iw^{([i])},$$ and we see that $\mathcal{\mathcal{C}}^{(q)}$ is generated by the basis $\{w, w^{([1])}, \dots, w^{([k-1])}\}$ as desired.

\end{proof}

\begin{lemma}\label{lem:finalGab}
 Let $\mathcal{C}$ be a linear MRD code of dimension $k<n$ with generator matrix
\[G^* = \left(\begin{array}{cccc}
g_1 & g_2 & \dots & g_n\\
g_1^{[s]} & g_2^{[s]}& \dots & g_n^{[s]}\\
&&\vdots\\
g_1^{[s(k-1)]} & g_2^{[s(k-1)]} & \dots & g_n^{[s(k-1)]}
              \end{array}
 \right) .\]
Then $g_1,\dots,g_n$ are linearly independent over $\F_q$.
\end{lemma}
\begin{proof}
 We prove this by contradiction. 
Assume that WLOG $g_1$ is in $\langle g_2,\dots,g_n\rangle_q$, i.e.\ there exist $\lambda_2,\dots,\lambda_n \in \F_q$ with $g_1 = \sum_{i=2}^n \lambda_i g_i$. Then 
\[g_1^{[j]} = \left(\sum_{i=2}^n \lambda_i g_i\right)^{[j]} = \sum_{i=2}^n \lambda_i^{[j]} g_i^{[j]}= \sum_{i=2}^n \lambda_i g_i^{[j]}\]
i.e.\ $g_1^{[j]} \in \langle g_2^{[j]},\dots,g_n^{[j]}\rangle_q$ for any $j\in \N$. 
Hence there exists $A\in \GL_n(q)$ such that the first column of $G^* A$ is zero. It follows from Theorem \ref{thm:mainMRD} that $\mathcal{C}$ is not a MRD code, which is a contradiction. 
%It follows that the first coordinate of all codewords is in the $q$-span of the other coordinates. Hence, the shortened code without the first coordinate has the same rank distance as the original code. On the other hand we know that, since the minimum distance of the code is $n-k+1\geq 2$, at least two of the $g_i$ are linearly independent over $\F_q$. By Lemma \ref{lem:indep} it follows that the shortened code has the same dimension as the original code. Then, by the Singleton bound, the shortened code (and thus the original code as well) has minimum rank distance at most $(n-1)-k+1=n-k$, which contradicts that the original code is MRD.
\end{proof}

\begin{theorem}\label{thm:mainGab}
Let $\mathcal{C}\subseteq \F_{q^m}^n$ be a linear MRD code of dimension $k<n$. Then $\dim(\mathcal{C} \cap \mathcal{C}^{[s]})=k-1$ if and only if $\mathcal{C}$ is a generalized Gabidulin code.
\end{theorem}
\begin{proof}
Let $\mathcal{C}$ be a generalized Gabidulin code of dimension $k$ with generalization parameter $s$. Then it follows from the structure of the generator matrix of $\mathcal{C}$ that $\dim(\mathcal{C} \cap \mathcal{C}^{[s]})=k-1$, which proves the first direction.

For the other direction we distinguish two cases:  
If $k\leq (n+1)/2$, then the minimum distance of $\mathcal{C}$ is at least $k$. Then it follows from Proposition \ref{prop:Gab} that $\mathcal{C}$ has a generator matrix of the form 
\[G^* = \left(\begin{array}{cccc}
g_1 & g_2 & \dots & g_n\\
g_1^{[s]} & g_2^{[s]}& \dots & g_n^{[s]}\\
&&\vdots\\
g_1^{[s(k-1)]} & g_2^{[s(k-1)]} & \dots & g_n^{[s(k-1)]}
              \end{array}
 \right) .\]
It follows from Lemma \ref{lem:finalGab} that the $g_i$ are linearly independent over $\F_q$. This is the definition of a generalized Gabidulin code.

If $k>(n+1)/2$, then  it follows from Proposition \ref{prop:dual1} that the dual code $\mathcal{C}^{\perp}\subseteq \F_{q^{m}}^{n}$ has dimension $n-k$ and minimum distance $k+1>n-k$, i.e.\ we can use Proposition \ref{prop:Gab} and Lemma  \ref{lem:finalGab} as before to show that $\mathcal{C}^{\perp}$ is a generalized Gabidulin code. Since the dual of a generalized Gabidulin code is again a generalized Gabidulin code (see Proposition \ref{prop:dual2}), the statement follows.
\end{proof}

%
%
%\vspace{3cm}
%
%\begin{lemma}
% Let $k=3$. Then Lemma 2 holds.
%\end{lemma}
%\begin{proof}
% There exist $h_1,h_2 \in \F_{q^m}^n$ such that $\mathcal{C}\cap \mathcal{C}^q = \langle h_1,h_2\rangle_{q^m}$. Since, $h_1,h_2\in \mathcal{C}^q$ we get that $h_1^{-q},h_2^{-q}\in \mathcal{C}$. Then we distinguish two cases:
%\begin{enumerate}
% \item If one of $h_i^{-q} \in \mathcal{C}\cap \mathcal{C}^q$, then $h_i^{-q^2}\in \mathcal{C}$ and (by Lemma 1) 
%\[G=\left(\begin{array}{c}h_i^{-q^2}\\ h_i^{-q}\\ h_i   \end{array}
%\right)\] 
%is a generator matrix of $\mathcal{C}$.
%\item If $h_1^{-q},h_2^{-q} \not \in \mathcal{C}\cap \mathcal{C}^q$, then
%\[\dim(\langle h_1^{-q},h_2^{-q} \rangle_{q^m} \cap (\mathcal{C} \cap \mathcal{C}^q) ) =1   \]
%\[\iff \dim(\langle h_1^{-q},h_2^{-q} \rangle_{q^m} \cap  \mathcal{C}^q ) =1   .\]
%\mathcal{C}hoose $f\neq 0$ in the intersection. Then $f\in \mathcal{C}$ and, since $f\in \mathcal{C}^q$, $f^{-q}\in \mathcal{C}$. Furthermore, since $f\in \langle  h_1^{-q},h_2^{-q} \rangle_{q^m}$, $f^q\in \mathcal{C}\cap \mathcal{C}^q$ and thus $f^q\in \mathcal{C}$. It follows that 
%\[G=\left(\begin{array}{c}f^{-q}\\ f\\f^q   \end{array}
%\right)\] 
%is a generator matrix of $\mathcal{C}$.
%\end{enumerate}
%
%\end{proof}
%

%%%%%%%%%%%%%%%%%%%%%%%%%%%%%%%%%%%%%%%%%%%%%%%%%%%%%%%%%%%%%%%%%%%%%%%%%%%%%%%%%%%%%%%%%%%%%%%%%%%%%%%%%%%%%%%%%%%%%%%%%%%%%%%%%%%%%%%%

\section{Non-Gabidulin MRD Codes}\label{sec:constructions}

\subsection{General results}

In this subsection we want to state some general results on the non-existence of non-Gabidulin MRD codes, i.e.\ for which parameters all MRD codes actually are Gabidulin codes.

\begin{theorem}\label{thm:trivdim}
 All linear MRD codes in $\F_{q^m}^n$ of dimension $k=1$ or $k=n-1$ are Gabidulin codes.
\end{theorem}
\begin{proof}
 Let $\mathcal{C}\subseteq \F_{q^m}^n$ be an MRD code of dimension $1$. Then the minimum rank distance is $n$ and it can be generated by one vector in $\F_{q^m}^n$. Clearly this vector needs to have only entries that are linearly independent over $\F_q$, thus it is a Gabidulin code.

Since the dual of a Gabidulin code is again a Gabidulin code (see Proposition \ref{prop:dual1}), the statement for codes of dimension $n-1$ follows.
\end{proof}

Then the following statement easily follows.

\begin{corollary}
 All linear MRD codes of length $n\in \{1,2,3\}$ are Gabidulin codes.
\end{corollary}

The following observation is helpful for further investigations:

\begin{lemma}\label{lem:helpMRD2}
 Any MRD code $\mathcal{C}\subseteq \F_{q^m}^n$ of dimension $k$ has a generator matrix $G\in \F_{q^m}^{k\times n}$ in systematic form, i.e.\ 
\[ G = \left[\; I_k \mid * \; \right] .\]
Moreover, all entries of $*$ are from $\F_{q^m}\backslash \F_q$.
\end{lemma}
\begin{proof}
The first statement is a direct consequence of Lemma \ref{lem:helpMRD}
% Assume that the generator matrix of $\mathcal{C}$ in reduced row echelon form has a row with pivot in column $i>k$. Then this row vector has at most $n-k$ many non-zero entries, which contradicts the minimum rank distance $n-k+1$ of $\mathcal{C}$. Therefore, all pivots are in columns $1,\dots, k$, which proves the first statement. 
The second statement follows from the minimum rank distance $n-k+1$ of the code, because every codeword needs to have at least $n-k$ entries from $\F_{q^m}\backslash \F_q$.
\end{proof}

In the first case not covered by Theorem \ref{thm:trivdim}, i.e.\ for length $n=4$ and dimension $k=2$, we can get the following statement. The same observation is mentioned as a computational result in \cite[Section V]{sh15}.

\begin{proposition}\label{prop:q=2}
 All linear MRD codes in $\F_{2^4}^4$ are Gabidulin codes.
\end{proposition}
\begin{proof}
 The case for codes of dimension $k=1$ or $k=3$ follows from Theorem \ref{thm:trivdim}. It remains to show the case $k=2$. Then by Lemma \ref{lem:helpMRD2} there exists a generator matrix of the form
\[G=\left( \begin{array}{cccc} 1&0&a&b\\ 0&1&c&d
           \end{array}
 \right)\]
with $a,b,c,d \in \F_{2^4}\backslash \F_2$. 
%Because of the minimum distance, which is $3$, $\{1,a,b\}, \{1,c,d\}, \{1,a,c\}, \{1,b,d\}$ must be linearly independent sets over $\F_2$, since the two rows of $G$ and their sum are codewords. 

By Theorem \ref{thm:mainMRD} a generator matrix $G$ of an MRD code satisfies 
\[ G \left( \begin{array}{cccc} 1&u_1&u_2&u_3\\ 0&1&u_4&u_5 \\ 0&0&1&u_6 \\ 0&0&0&1
           \end{array}
 \right)  =  \left( \begin{array}{cccc} 1&u_1&u_2+a&u_3+au_6+b\\ 0&1&u_4+c&u_5+cu_6+d
           \end{array}
 \right)\]
needs to have only non-zero maximal minors for $u_1,\dots,u_6 \in \F_2$. Thus we get the following inequations:
\begin{align*}
1&\neq 0\\
 u_4+c &\neq 0 \\
u_5+cu_6 +d &\neq 0 \\
(u_2+a) + u_1(u_4+c) &\neq 0 \\
(u_3+au_6+b) + u_1(u_5+cu_6+d)&\neq 0\\
(u_2+a)(u_5+cu_6+d) + (u_4+c)(u_3+ au_6+b) &\neq 0   .
\end{align*}
Clearly the first inequation is always true; the same for the second, since $u_4\in \F_2$ and $c\notin \F_2$.
%One can easily check that the first four inequalitions are always fulfilled by the conditions from above on $a,b,c,d \not\in \F_2$ and $u_1,\dots,u_6\in \F_2$. Similarly we can rewrite the fifth inequality as
%\[u_3+u_1u_5   \neq u_6(u_1c +a)  + (u_1d+b) .\]
%This is also always true, since the left side is in $\F_2$ and the right side is not in $\F_2$. This can be seen by the following observation: $(\; +1, u_1, u_1c+a, u_1d+b\;)$ is a codeword and must hence have rank $3$ over $\F_2$; which implies that $(u_1c+a)$ and $(u_1d+b)$ are linearly independent over $\F_2$.
%Therefore it remains only the sixth inequality to be checked for the MRD property.

If $G$ does not generate a Gabidulin code then, by Theorem \ref{thm:mainGab}, 
\[\rk \left[\begin{array}{cccc}
             1&0&a&b \\
	    0&1&c&d\\
 1&0&a^2&b^2 \\
	    0&1&c^2&d^2
            \end{array}
\right] \neq 3 .\]
Since $a,b,c,d \not \in \F_2$ the rank of the above matrix is at least $3$. Thus we need that the rank is equal to $4$, which is equivalent to 
\[(a^2+a)(d^2+d) + (b^2+b)(c^2+c) \neq 0 .\]

Thus, overall, we need to check that there is no solution to the system of inequations
\begin{align*}
%\lambda_1 a + \lambda_2 b + \lambda_3 &\neq 0 \\
%\lambda_4 c + \lambda_5 d + \lambda_6 &\neq 0 \\
%\lambda_7 a + \lambda_8 c + \lambda_9 &\neq 0 \\
%\lambda_{10} b + \lambda_{11} d + \lambda_{12} &\neq 0 \\
u_5+cu_6 +d \neq 0 \\
(u_2+a) + u_1(u_4+c) \neq 0 \\
(u_3+au_6+b) + u_1(u_5+cu_6+d)\neq 0\\
(u_2+a)(u_5+cu_6+d) + (u_4+c)(u_3+ au_6+b) \neq 0 \\
(a^2+a)(d^2+d) + (b^2+b)(c^2+c) \neq 0 
\end{align*}
for any $u_1,\dots,u_6\in\F_2$. 
With the help of a computer program one can check that there exists no solution for $a,b,c,d \in \F_{2^4}\backslash \F_2$ for the above system of inequations, for any representation of the extension field.
\end{proof}

The previous results show that the first set of parameters for which we can hope to construct non-Gabidulin MRD codes is $n=4, k=2$ and $q\geq 3$. This is what we will do in the following subsection.

%%%%%%%%%%%%%%%%%%%%%%%%%%%%%%%%%

\subsection{Constructions of length $4$ and dimension $2$}

In this subsection we use the results of the previous sections to derive some linear MRD codes that are not generalized Gabidulin codes. The codes that we derive in this subsection have length $4$ and dimension $2$.

\begin{theorem}\label{thm:construction}
 Let $m>4$,  $\alpha \in \F_{q^m}$ primitive such that $\F_{q^m}^* = \langle \alpha \rangle$ and $\gamma \in \F_{q}$ be 
 a quadratic non-residue in $\F_q$ 
 such that $\gamma \neq  (\alpha^{[s]}+\alpha)^2$ for any $0<s<m$ with $\gcd (s,m)=1$. Then
\[G = \left( \begin{array}{cccc}
             1&0&\alpha&\alpha^2 \\
	    0&1&\alpha^2 & \gamma \alpha
            \end{array}\right)
\]
is a generator matrix of an MRD code $\mathcal{C}\subseteq \F_{q^m}^4$ of dimension $k=2$ that is not a generalized Gabidulin code.
\end{theorem}
\begin{proof}
 First we prove that $\mathcal{C}$ is MRD. For this we use Corollary \ref{cor:mainMRD}. Note that 
\[\mathrm{UT^*}_4(q) = \left\{
\left( \begin{array}{cccc} 1&u_1&u_{2}&u_{3} \\ 0&1&u_{4}&u_{5 }\\ 0&0&1&u_{6} \\ 0&0&0&1
       \end{array}
 \right) \bigg \vert \;  u_{1},\dots,u_{6} \in \F_q
 \right\}\]
and
\[G \left( \begin{array}{cccc} 1&u_1&u_{2}&u_{3}\\ 0&1&u_{4}&u_{5} \\ 0&0&1&u_{6} \\ 0&0&0&1
       \end{array}
 \right)  = \left( \begin{array}{cccc} 1&u_1&u_{2}+\alpha & u_{3}+u_{6}\alpha+ \alpha^2 \\ 0&1&u_{4}+ \alpha^2 &  u_{5}+u_{6}\alpha^2+\gamma\alpha   \end{array}
 \right).\]
We need to show that all maximal minors of this matrix are non-zero for any values of $u_{1},\dots,u_{6}$:
\begin{align*}
1&\neq 0\\
 u_4+\alpha^2 &\neq 0 \\
u_5+\alpha^2 u_6 +\gamma\alpha &\neq 0 \\
(u_2+\alpha) - u_1(u_4+\alpha^2) &\neq 0 \\
(u_3+\alpha u_6+\alpha^2) - u_1(u_5+\alpha^2 u_6+\gamma \alpha)  &\neq 0\\
(u_2+\alpha)(u_5+\alpha^2 u_6+\gamma \alpha) - (u_4+\alpha^2)(u_3+ \alpha u_6+\alpha^2) &\neq 0 .
\end{align*}
One can easily see that the first four inequations are always true, since all $u_i$ are in $\F_q$. We can rewrite the fifth inequation as
\[(u_1 u_5 - u_3) + (u_1 \gamma -u_6)\alpha + (u_1u_6 -1)\alpha^2 \neq 0 .\]
If the last term is zero then $u_1 = u_6^{-1}$. But then $u_1 \gamma - u_6 = u_6^{-1} (\gamma - u_6^2) \neq 0$ because $\gamma$ is a quadratic non-residue. Thus, in this case, the middle term of the above sum does not vanish, i.e.\ the inequation is always true. 
%The five minors involving the first or second column are non-zero, similarly to the proof of Proposition \ref{prop:q=2}.
%since $u_{4}+\alpha^2$ and $u_{5}+u_{6}\alpha^2+\gamma\alpha$ are non-zero. Similarly the minors involving the second column are non-zero, since $b+\alpha$ and $c+f\alpha + \alpha^2$ are non-zero. 
%Remains the minor involving the third and fourth column:
Lastly we can rewrite the sixth inequation as 
%\[(u_{2}+\alpha)(u_{5}+u_{6}\alpha^2+\gamma\alpha) - (u_{4}+\alpha^2) (u_{3}+u_{6}\alpha + \alpha^2)= \]
\[ (u_{2}u_{5}-u_{3}u_{4}) + (u_{2}\gamma + u_{5} -u_{4}u_{6})\alpha + (u_{2}u_{6}+\gamma-u_{4}-u_{3})\alpha^2 - \alpha^4 \neq 0 .\]
This is always true, since the minimal polynomial of $\alpha$ has degree $m>4$ and $u_{1},\dots, u_{6},\gamma \in \F_q$, i.e.\ nothing can cancel out the $\alpha^4$-term.

It remains to prove that $\mathcal{C}$ is not a generalized Gabidulin code. For this we use Theorem \ref{thm:mainGab} and compute
\[\rk\left[\begin{array}{c} G\\G^{[s]}   \end{array}\right] =
\rk \left[\begin{array}{cccc}
             1&0&\alpha&\alpha^2 \\
	    0&1&\alpha^2 & \gamma \alpha\\
 1&0&\alpha^{[s]}&\alpha^{2[s]} \\
	    0&1&\alpha^{2[s]} & \gamma \alpha^{[s]}
            \end{array}
\right] =
\]
\[
\rk \left[\begin{array}{cccc}
             1&0&\alpha&\alpha^2 \\
	    0&1&\alpha^2 & \gamma \alpha\\
 0&0&\alpha^{[s]}-\alpha&\alpha^{2[s]}-\alpha^2 \\
	    0&0&\alpha^{2[s]}-\alpha^2 & \gamma (\alpha^{[s]} - \alpha)
            \end{array}
\right] ,
\]
for any $s$ with $\gcd(s,m)=1$. 
Since $\alpha \not \in \F_{q}$ this rank cannot be equal to $2$, by Lemma \ref{lem:help}. Hence, $\mathcal{C}$ is Gabidulin if and only if the determinant of the lower right submatrix from above is zero, i.e.\ if and only if
\[\gamma (\alpha^{[s]}-\alpha)^2 - (\alpha^{2[s]}-\alpha^2)^2 =0   \]
\[\iff \gamma (\alpha^{[s]}-\alpha)^2  = (\alpha^{2[s]}-\alpha^2)^2\]
\[\iff \gamma (\alpha^{[s]}-\alpha)^2  = (\alpha^{[s]}-\alpha)^2(\alpha^{[s]}+\alpha)^2\]
\[\iff \gamma   = (\alpha^{[s]}+\alpha)^2 .\]
%Since $\gamma$ is a quadratic non-residue this is never possible, 
This is a contradiction to the conditions on $\gamma$, 
which implies that $\mathcal{C}$ is non-Gabidulin.
\end{proof}

Note that in the previous theorem $\gamma \in \F_{q}$ can in particular be chosen as a quadratic non-residue in the extension field $\F_{q^{m}}$.

\begin{example}
 Let $q=3, m=5$ and $\alpha$ a root of $x^5 + 2x^2+ x + 1$.  Then $\gamma = 2$ is a non-quadratic residue in $\F_{3^5}$ and the code with generator matrix
\[G = \left( \begin{array}{cccc}
             1&0&\alpha&\alpha^2 \\
	    0&1&\alpha^2 & 2 \alpha
            \end{array}\right)
\]
is an MRD but not a generalized Gabidulin code. 
\end{example}

Although we proved Theorem \ref{thm:construction} for $m>4$ we can find analog constructions for $m=4$, as shown in the following examples. The proof that these examples are also non-Gabidulin MRD codes is analogous to the one of Theorem \ref{thm:construction}, but when checking if the maximal minor of $G\mathrm{UT}^{*}_{n}(q)$ involving the third and fourth column is non-zero we cannot use the argument that the minimal polynomial $m(x)$ of $\alpha$ has degree at least $4$.  Instead we need to write $\alpha^{4}$ modulo $m(x)$ and show that the minor is non-zero. 
%Moreover, we do not use quadratic non-residues for the values of $\gamma$.

\begin{example}\label{ex1}
Let $q=3, m=4$ and $\alpha$ a root of $x^4 -x^3 - 1$.  %The only quadratic non-residue in $\F_{3}$ is $\gamma=2$, but this is not a quadratic non-residue in $\F_{3^{4}}$. 
Then $\gamma=2$ is a quadratic non-residue in $\F_3$ and it fulfills the conditions that $\gamma \neq  (\alpha^{[s]}+\alpha)^2$ for any $0<s<m$ with $\gcd (s,m)=1$. 
Now the code with generator matrix
\[G = \left( \begin{array}{cccc}
             1&0&\alpha&\alpha^2 \\
	    0&1&\alpha^2 & 2 \alpha
            \end{array}\right)
\]
is an MRD but not a generalized Gabidulin code.
To show that it is an MRD code we need to prove that the before mentioned minor is non-zero, i.e.\ that
\[(u_{2}u_{5}-u_{3}u_{4}) + (2u_{2}   + u_{5} -u_{4}u_{6})\alpha + (u_{2}u_{6}+ 2 -u_{4}-u_{3})\alpha^2 - \alpha^4\]
\[\iff (u_{2}u_{5}-u_{3}u_{4}-1) + (2u_{2}   + u_{5} -u_{4}u_{6})\alpha + (u_{2}u_{6}+ 2 -u_{4}-u_{3})\alpha^2 - \alpha^{3}  \]
is non-zero for any $u_{1},\dots, u_{6}\in \F_{q}$. This is clearly the case since nothing can cancel out the $\alpha^{3}$-term.
%
%To finally show that it is not a generalized Gabidulin code (with $s=1,3$) we compute
%\[(\alpha^3-\alpha)(2\alpha^3-2\alpha)-(\alpha^6-\alpha^2)^2 = 2\alpha^3 + 2\alpha^2 + 2\alpha \neq 0 ,\]
%\[(\alpha^{27}-\alpha)(2\alpha^{27}-2\alpha)-(\alpha^{54}-\alpha^2)^2 = 2\alpha^2+1  \neq 0 ,\]
%which proves the statement.
\end{example}

Note that in the previous example we could have chosen any minimal polynomial of $\alpha$ that involves a non-zero term of order $3$ (and a suitable $\gamma$). The same proof would then show that the generated code is MRD but not a generalized Gabidulin code.

We want to conclude with a final example over $\F_{5}$. A generalization for other values of $q$ is then straight-forward.
\begin{example}
Let $q=5, m=4$ and $\alpha$ a root of $x^4 +x^3+x^{2}+x+3$.   
Then $\gamma=2$ is a quadratic non-residue in $\F_5$ and it fulfills the conditions that $\gamma \neq  (\alpha^{[s]}+\alpha)^2$ for any $0<s<m$ with $\gcd (s,m)=1$. 
Now the code with generator matrix
\[G = \left( \begin{array}{cccc}
             1&0&\alpha&\alpha^2 \\
	    0&1&\alpha^2 & 2 \alpha
            \end{array}\right)
\]
is an MRD but not a generalized Gabidulin code.
To show that it is an MRD code we need to prove that the before mentioned minor is non-zero, i.e.\ that
\[(u_{2}u_{5}-u_{3}u_{4}) + (u_{2}\gamma + u_{5} -u_{4}u_{6})\alpha + (u_{2}u_{6}+\gamma-u_{4}-u_{3})\alpha^2 - \alpha^4\]
\[\iff (u_{2}u_{5}-u_{3}u_{4}+2) + (u_{2}\gamma + u_{5} -u_{4}u_{6}+4)\alpha + (u_{2}u_{6}+\gamma-u_{4}-u_{3}+4)\alpha^2 + 4\alpha^{3}  \]
is non-zero for any $u_{1},\dots, u_{6}\in \F_{q}$. This is clearly the case since nothing can cancel out the $4\alpha^{3}$-term.
%
%To finally show that it is not a generalized Gabidulin code (with $s=1,3$) we compute
%\[(\alpha^5-\alpha)(2\alpha^5-2\alpha)-(\alpha^{10}-\alpha^2)^2 = 2- (\alpha^5+\alpha)^2 = \alpha^3 + 4\alpha^2 + 4\alpha + 4  \neq 0 ,\]
%\[(\alpha^{125}-\alpha)(2\alpha^{125}-2\alpha)-(\alpha^{250}-\alpha^2)^2 = 3\alpha^3 + 4\alpha^2 + 3\alpha + 3 \neq 0 ,\]
%which proves the statement.
\end{example}

%Note that the same remark as for Example \ref{ex1} holds, i.e.\ that we could have chosen any minimal polynomial of $\alpha$ that involves a non-zero term of order $3$.

%%%%%%%%%%%%%%%%%%%%%%%%%%%%%%%%%

\subsection{Construction of length $5$ and dimension $2$}

Analogously to the previous subsection, we present a construction of linear MRD codes of length $5$ and dimension $2$ that are not generalized Gabidulin codes. 

\begin{theorem}\label{thm:construction5}
 Let $m>7$,  $\alpha \in \F_{q^m}$ primitive such that $\F_{q^m}^* = \langle \alpha \rangle$ and $\gamma \in \F_{q}$ be 
 such that $\gamma \neq  (\alpha^{[s]}+\alpha)(\alpha^{2[s]} + \alpha^{[s]+1} +\alpha^2)$ for any $0<s<m$ with $\gcd (s,m)=1$. Then
\[G = \left( \begin{array}{ccccc}
             1&0&\alpha&\alpha^2&\alpha^3 \\
	    0&1&\alpha^2 &\alpha^4& \gamma \alpha
            \end{array}\right)
\]
is a generator matrix of an MRD code $\mathcal{C}\subseteq \F_{q^m}^5$ of dimension $k=2$ that is not a generalized Gabidulin code.
\end{theorem}
\begin{proof}
 First we prove that $\mathcal{C}$ is MRD. For this we use Corollary \ref{cor:mainMRD}. Note that 
\[\mathrm{UT^*}_5(q) = \left\{
\left( \begin{array}{ccccc} 1&u_1&u_{2}&u_{3}&u_{4} \\ 0&1&u_{5}&u_{6 }&u_{7}\\ 0&0&1&u_{8}&u_{9} \\ 0&0&0&1&u_{10}\\0&0&0&0&1
       \end{array}
 \right) \bigg \vert \;  u_{1},\dots,u_{10} \in \F_q
 \right\}\]
and
\[G \left( \begin{array}{ccccc} 1&u_1&u_{2}&u_{3}&u_{4} \\ 0&1&u_{5}&u_{6 }&u_{7}\\ 0&0&1&u_{8}&u_{9} \\ 0&0&0&1&u_{10}\\0&0&0&0&1
       \end{array}
 \right) = \]\[
 \left( \begin{array}{ccccc} 1&u_1&u_{2}+\alpha & u_{3}+u_{8}\alpha+ \alpha^2 & u_4+u_9\alpha+u_{10}\alpha^2 +\alpha^3 \\ 0&1&u_{5}+ \alpha^2 &  u_{6}+u_{8}\alpha^2+\alpha^4 & u_7+u_9 \alpha^2 +u_{10}\alpha^4 + \gamma \alpha   \end{array}
 \right).\]
We need to show that all maximal minors of this matrix are non-zero for any values of $u_{1},\dots,u_{10}$. Analogously to the proof of Theorem \ref{thm:construction}, one can easily see that the minors involving the first column are non-zero. The same holds for the minor involving the second and third column. The following equations remain:
\begin{eqnarray}
%(u_2+\alpha) - u_1(u_5+\alpha^2) &\neq 0 \\
(u_3+\alpha u_8+\alpha^2) - u_1(u_6+\alpha^2 u_8+\alpha^4)  &\neq 0\\
(u_4+\alpha u_9+\alpha^2 u_{10} + \alpha^3) - u_1(u_7+\alpha^2 u_9+\alpha^4 u_{10}+ \gamma \alpha)  &\neq 0\\
(u_2+\alpha)(u_6+\alpha^2 u_8+ \alpha^4 u_{10} + \alpha^4) - (u_5+\alpha^2)(u_3+ \alpha u_8+\alpha^2) &\neq 0 \\
\nonumber(u_2+\alpha)(u_7+\alpha^2 u_9+ \alpha^4 u_{10} + \gamma \alpha) -\\  (u_5+\alpha^2)(u_4+ \alpha u_9+\alpha^2 u_{10} + \alpha^3) &\neq 0 \\
\nonumber( u_{3}+u_{8}\alpha+ \alpha^2)(u_7+u_9 \alpha^2 +u_{10}\alpha^4 + \gamma \alpha ) -\\ (u_4+u_9\alpha+u_{10}\alpha^2 +\alpha^3)( u_{6}+u_{8}\alpha^2+\alpha^4) &\neq 0
\end{eqnarray}
We can rewrite Inequation (1) as 
$$ (u_3-u_1 u_6) + u_8 \alpha + (1-u_1 u_8)\alpha^2 - u_1 \alpha^4\neq 0 .$$
The $\alpha^2$-term only vanishes if $u_1 u_8 =1$, but then the $\alpha$-term (and the $\alpha^4$-term) do not vanish. Hence, this inequation is always true. Inequation (2) has an $\alpha^3$-term that never vanishes, thus it is also true. Similarly, Inequation (3) has an $\alpha^5$-term that never vanishes, and Inequation (7) has an $\alpha^7$-term, that never vanishes. These two inequations are therefore also true. We can rewrite Inequation (4) as
$$ (u_2u_7 - u_4 u_5) + (u_2 \gamma + u_7 -u_5 u_9) \alpha + (u_2 u_9 +\gamma - u_5 u_{10} -u_4) \alpha^2  -u_5 \alpha^3 $$
$$+ (u_2 u_{10} -u_{10}) \alpha^4 + (u_{10} -1) \alpha^5 \neq 0 .$$
For the $\alpha^5$-term to vanish we need $u_{10}=1$, for the $\alpha^3$-term to vanish we need $u_{5}=0$. If additionaly we want the $\alpha^4$-term to vanish we need $u_{2}=1$. Then we need $u_7=0$ for the first summand to be zero. But then the $\alpha$-term does not vanish, since $\gamma\neq 0$. Thus this inequation is also true. Therefore we have shown that $\mathcal C $ is an MRD code.

It remains to prove that $\mathcal{C}$ is not a generalized Gabidulin code. For this we use Theorem \ref{thm:mainGab} and compute
\[\rk\left[\begin{array}{c} G\\G^{[s]}   \end{array}\right] =
\rk \left[\begin{array}{ccccc}
             1&0&\alpha&\alpha^2 &\alpha^3 \\
	    0&1&\alpha^2 &\alpha^4& \gamma \alpha\\
 1&0&\alpha^{[s]}&\alpha^{2[s]} &\alpha^{3[s]} \\
	    0&1&\alpha^{2[s]} &\alpha^{4[s]} & \gamma \alpha^{[s]}
            \end{array}
\right] =
\]
\begin{equation}\label{eq1}
\rk \left[\begin{array}{ccccc}
             1&0&\alpha&\alpha^2&\alpha^3 \\
	    0&1&\alpha^2 &\alpha^4& \gamma \alpha\\
 0&0&\alpha^{[s]}-\alpha&\alpha^{2[s]}-\alpha^2&\alpha^{3[s]}-\alpha^3 \\
	    0&0&\alpha^{2[s]}-\alpha^2 &\alpha^{4[s]}-\alpha^4& \gamma (\alpha^{[s]} - \alpha)
            \end{array}
\right] ,
\end{equation}
for any $s$ with $\gcd(s,m)=1$. 
Since $\alpha \not \in \F_{q}$ this rank cannot be equal to $2$, by Lemma \ref{lem:help}. Hence, $\mathcal{C}$ is Gabidulin if and only if the rank of the matrix in \eqref{eq1} is equal to $3$. We compute the determinant of the lower submatrix involving columns $3$ and $5$,
\[\gamma (\alpha^{[s]}-\alpha)^2 - (\alpha^{2[s]}-\alpha^2)(\alpha^{3[s]}-\alpha^3)  =  \]
\[(\alpha^{[s]}-\alpha)^2 (\gamma - (\alpha^{[s]}+\alpha)(\alpha^{2[s]}-\alpha^{[s]+1} + \alpha^2) )  , \]
which is non-zero by the conditions on $\gamma$. Hence the rank of the matrix from \eqref{eq1} is $4$, 
which implies that $\mathcal{C}$ is non-Gabidulin.
\end{proof}

\begin{example}
 Let $q=2, m=8$ and $\alpha$ a root of $x^8 +x^4+x^3+x^{2}+1$.   
Then $\gamma=1$ fulfills the conditions that $\gamma \neq  (\alpha^{[s]}+\alpha)(\alpha^{2[s]} + \alpha^{[s]+1} +\alpha^2)$ for any $0<s<m$ with $\gcd (s,m)=1$. 
Now the code with generator matrix
\[G = \left( \begin{array}{ccccc}
             1&0&\alpha&\alpha^2& \alpha^3 \\
	    0&1&\alpha^2 & \alpha^4 &  \alpha
            \end{array}\right)
\]
is an MRD but not a generalized Gabidulin code.
\end{example}

Note that, analogously to the constructions of length $4$ from the previous subsection, one can use the construction from Theorem \ref{thm:construction5} to construct non-Gabidulin MRD codes, also if $5\leq m\leq 7$. One simply needs to check that the minimal polynomial of $\alpha$ is such that all the inequations arising from $G\; \mathrm{UT^*}_5(q)$ hold.

%%%%%%%%%%%%%%%%%%%%%%%%%%%%%%%%%%%%%%%%%%%%%%%%%%%%%%%%%%%%%%%%%%%%%%%%%%%%%%%%%%%%%%%%%%%%%%%%%%%%%%%%%%%%%%%%%%%%%%%%%%%%%%%%%%%%%%%%

\section{Conclusion}\label{sec:conclusion}

In this work we give a new criterion to check if a given matrix generates a linear MRD code. Moreover, we derive a criterion to check if a given generator matrix belongs to a linear generalized Gabidulin code or not. Although the criterion itself is quite simple, the proof of it involves several, to our knowledge new, technical lemmas on the Frobenius map, as well as the $\F_q$-rank and linear independence of elements in $\F_{q^m}$. 

We then use these results to construct linear MRD codes that are not generalized Gabidulin codes. Since the class of Gabidulin codes is closed under the semilinear isometries (also called equivalencies by some authors) this means that these codes are also not semilinearly isometric (or equivalent) to generalized Gabidulin codes.

In future work we want to use these criteria to find more general constructions for non-Gabidulin MRD codes. Moreover, we would like to classify all linear MRD codes and see how many different classes of codes there are, for given parameter sets.

We also believe that the results of this paper are interesting from a cryptographic point of view. Especially for the cryptanalysis of McEliece-type cryptosystems based on Gabidulin codes  \cite{Gabidulin91,Loidreau10} our criteria for MRD and Gabidulin codes might lead to new and more efficient attacks. First results in this direction can be found in \cite{ho15a}, and we would like to pursue this line of research further in the future.

\section*{Acknowledgement}

The authors would like to thank Heide Gluesing-Luerssen and Joachim Rosenthal for their valuable comments and advise on this work.

\bibliographystyle{plain}
\bibliography{network_coding_stuff,biblio}

\end{document}